\documentclass[
  UKenglish,
  envcountsect,
  envcountsame,
]{llncs}

\usepackage{booktabs}
\usepackage{tikz}

\usepackage{listings}
\lstdefinelanguage{pseudocode-bannachtantauskambath}{
  morekeywords={
    algorithm,each,if,then,else,while,do,repeat,until,seq,seqdo,return,call,
    for,pardo,in,foreach,print,output,input,exit,
    break,loop,end,begin,goto,pardo,par,global,read,write,to,stop,idle,procedure,function,
    parallel,
  },
  sensitive=true,
  morecomment=[l]{//},
  morestring=[b]",
}

\lstdefinestyle{pseudocode-bannachtantauskambath}{
  language=pseudocode-bannachtantauskambath,
  numberstyle  =\scriptsize,
  numbers      =left,
  basicstyle   =\rmfamily,
  keywordstyle =\bfseries\itshape,
  columns      =fullflexible,
  mathescape   =true,
  identifierstyle=\itshape,
  literate={<}{{$<$}}1
  {>}{{$>$}}1
  {-}{{-}}1
  {+}{{$+$}}1
  {*}{{$\cdot$}}1
  {<-}{{$\gets$\ }}2
  {<=}{{$\leq$}}1
  {>=}{{$\geq$}}1
  {!=}{{$\neq$}}1
  {==}{{$=$}}1
}

\usepackage{amsmath,amsfonts}

\newenvironment{claimagain}[1]{\par\medskip\noindent\textbf{\bfseries #1.}\ \ \bgroup\itshape\ignorespaces}{\egroup}
\let\origlemma\lemma
\def\lemma#1#2#3\end{\origlemma%
  \expandafter\gdef\csname savedlemma#2\endcsname{#3}%
  \label{#2}#3\end}
\def\reclaimlemma#1{\claimagain{Claim of Lemma~\ref{#1}}\expandafter\expandafter\expandafter\ignorespaces\csname savedlemma#1\endcsname\endclaimagain\ignorespaces}
\let\origtheorem\theorem
\def\theorem#1#2#3\end{\origtheorem%
  \expandafter\gdef\csname savedtheorem#2\endcsname{#3}%
  \label{#2}#3\end}
\def\reclaimtheorem#1{\claimagain{Claim of Theorem~\ref{#1}}\expandafter\expandafter\expandafter\ignorespaces\csname savedtheorem#1\endcsname\endclaimagain\ignorespaces}

\newtheorem{fact}[theorem]{Fact}
\newtheorem{observation}[theorem]{Observation}

\newcommand\Class[1]{%
  \mathchoice%
  {\text{\normalfont\fontsize{9pt}{10pt}\selectfont$\mathrm{#1}$}}%
  {\text{\normalfont\fontsize{9pt}{10pt}\selectfont$\mathrm{#1}$}}%
  {\text{\normalfont$\mathrm{#1}$}}%
  {\text{\normalfont$\mathrm{#1}$}}%
}

\newcommand{\Lang}[1]{%
  \ifmmode{%
    \text{\normalfont\textsc{#1}}%
  }%
  \else
  {\normalfont\textsc{#1}}%
  \fi}

\newcommand\Algo[1]{%
  \mathchoice%
  {\text{\normalfont\textsf{#1}}}%
  {\text{\normalfont\textsf{#1}}}%
  {\text{\normalfont\textsf{#1}}}%
  {\text{\normalfont\textsf{#1}}}%
}

\newcommand\PLang[1][]{p\def\test{#1}\ifx\test\stockhustantauempty\else_{\mathit{#1}}\fi\text-\penalty15\Lang}
\newcommand\PLangText[1][]{$p\def\test{#1}\ifx\test\empty\else_{\mathrm{#1}}\penalty15\fi$-\penalty15\hskip0pt\textsc}

\title{Towards Work-Efficient Parallel~Parameterized~Algorithms}

\author{%
  Max Bannach\inst{1} \and
  Malte Skambath\inst{2} \and
  Till Tantau\inst{1}}
\authorrunning{M. Bannach et al.}

\institute{
  Institute for Theoretical Computer Science, Universit\"at zu L\"ubeck, Germany\\
  \email{\{bannach,tantau\}@tcs.uni-luebeck.de}%
  \and
  Department of Computer Science, Kiel University, Germany\\
  \email{malte.skambath@email.uni-kiel.de}%
}

\begin{document}
\maketitle

\begin{abstract}
  Parallel parameterized complexity theory studies how 
  fixed-parameter tractable (fpt) problems can be solved in
  parallel. Previous theoretical work focused on parallel
  algorithms that are very fast in principle, but did not take into
  account that when we only have a small number of processors
  (between 2 and, say,~1024), it is more important
  that the parallel algorithms are \emph{work-efficient}. In the
  present paper we investigate how work-efficient fpt algorithms can
  be designed. We review standard methods from fpt theory, like
  kernelization, search trees, and interleaving, and prove
  trade-offs for them between work efficiency and runtime
  improvements. This results in a toolbox for
  developing work-efficient parallel fpt algorithms.

\end{abstract}

\section{Introduction}

Since its introduction by Downey and Fellows~\cite{DowneyF1999} about thirty years ago,
parameterized complexity theory has been successful at
identifying which problems are fixed-parameter tractable (fpt), but has
also had high practical impact. Efforts to formalize and devise
\emph{parallel} fpt algorithms date back twenty
years~\cite{CesatiI1998,CheethamDRST03}, but a lot of the theoretical
research is quite recent~\cite{BannachST15,BannachT17,ElberfeldST15}.
The findings can be summarized, very
briefly, as follows: (1) It is possible to classify the problems in
the class $\Class{FPT}$ of fixed-parameter tractable problems
according to how well they can be solved in parallel. (2) We find
natural parameterized problems on all levels, from 
problems in $\Class{FPT}$ that are inherently sequential to problems
that can be solved in constant (!) parallel time.

One aspect that the existing research lacks
-- and which may also
explain the small number of actual implementations --
is a fine-grained analysis of the \emph{work} done by parallel fpt
algorithms, which is defined as the total number of
computational steps done by an algorithm summed over all processing units
(in particular, for a sequential algorithm, its work equals its
runtime). Unfortunately, ``the work must be done'': on a machine with
$p$ processors, a parallel algorithm with $W(n)$ work cannot
finish faster than in time $W(n) / p$ on length-$n$ inputs. Since 
real-life values of~$p$ are small (between 2 and perhaps~1024), a
large $W(n)$ can lead to actual runtimes (``wall clock runtimes'')
that are larger than those of sequential algorithms. 

A common pattern in the design and analysis of parallel
algorithms is that as we try to decrease the work $W(n)$ in order to
get down the quotient $W(n) / p$, the ``theoretical'' parallel runtime
$T(n)$ rises. This pattern is repeated in the fpt setting:
Table~\ref{table:vertex-cover} shows the work and time needed by
different parallel algorithms for $\PLang{vertex-cover}$.
Note that we will never be able to reduce the work of a parallel algorithm
below the work of the fastest sequential algorithm and an
algorithm is called \emph{work-optimal} if it matches this lower bound.

\begin{table}[t]
  \caption{Faster parallel algorithms for
    $\PLang{}$\textsc{vertex-cover} entail more work. We can achieve a
    runtime of $O(1)$ at the cost 
    of the expensive use of color coding~\cite{BannachST15}.  If we
    allow $O(\log n)$ time, a parallel Buss kernelization in
    conjunction with a simple brute force algorithm reduces the
    work. The next two lines are based on shallow search trees,
    discussed in Section~\ref{sec:shallow-st}, and the work starts to
    become competitive with sequential algorithms. The last lines show
    that being work-competitive to the best known sequential
    algorithms implies larger and larger runtimes.}
  \label{table:vertex-cover}
  \begin{tabular}{ll}
    \toprule
    Work                & Parallel Time  \\\midrule
    $O(kn + 2^{2^k+k})$ & $O(1)$  \\
    $O(kn + 2^{k^2}\cdot k^2)$ \qquad\hbox{}& $O(\log n)$ \\
    $O(kn + 3^{k}k^2)$           & $O(\log n+\log^2 (k))$ \\
    $O(kn + 2^{k})$           & $O(\log n + \log^4 (k))$  \\
    $O(kn + 1.6181^{k})$       & $O(\log n +  k\log(k))$  \\
    $O(kn + 1.4656^{k})$    & $O(\log n +  k\log(k))$     \\
    $O(kn + 1.2738^{k})$    & $O(\log n + k^4\sqrt{k}))$  \\
    \bottomrule
  \end{tabular} 
\end{table}

\paragraph*{Our Contributions.}  
Many fpt algorithms are based on the \emph{search
  tree technique,} which recursively 
traverses a search tree whose depth and degree are
bounded by the parameter, resulting in a sequential runtime of the
form $c^k$ or perhaps $(ck)^k$ for some constant~$c$. Intuitively, search tree
algorithms should be easy to parallelize since the different branches 
of the search tree can be processed independently.
We show that this intuition is correct and we provide precise
conditions for search tree algorithms under which they can be turned into
work-efficient parallel algorithms.
A parallel search tree algorithm still has to process, and thus
construct, all branches of the tree, leading to a
parallel runtime that is proportional to the depth of the search tree,
which is normally $\Theta(k)$. This theoretical runtime is
typically much smaller than the actual wall-clock time
$W(n)/p=\bigl(c^k+O(n)\bigr)/p >\!\!> k$.
However, we show that in some cases there is room for improvement and
the runtime of $\Omega(k)$ can be replaced by $O(\log k)$ without
increasing the work. The idea is to modify the search tree such that
it ``branches aggressively,'' thereby reducing the depth to $O(\log k)$.

A second tool of parameterized complexity theory are
\emph{kernelizations:} mappings from input instances to membership-equivalent
instances whose size is bounded by the parameter.
Some problems admit more than one kernelization 
and we may be able to speedup the computation by applying all of them
in a cleverly chosen order. For sequential computations, determining this order
is simple: First apply the fastest kernel, which 
may however result in a still rather large instance. Then apply a
slower kernel with a smaller output -- the high runtime 
matters less since it is applied to a smaller
input. Such \emph{kernel cascades} are also possible in the parallel
setting, but here kernelizations may have incomparable work, runtime,
and output size. We provide a general procedure to combine a set of
parallel kernelizations into a work-efficient and fast kernelization
that minimizes the output size.

A third tool is \emph{interleaving:} Instead of using a kernelization
just as a preprocessing procedure, during a search tree traversal 
call the kernel algorithm at each tree node to ensure that the
intermediate instances are small. In the sequential setting this has
the desirable effect  
of turning a runtime of the form $O(k^c\cdot \xi^k+n^c)$ into one
of the form $O(\xi^k+n^c)$~\cite{NiedermeierR2000}. We show that interleaving is also
possible in the parallel setting in a work-efficient manner, including
the mentioned depth-$O(\log k)$ search trees that do not
arise in the sequential setting. 

\paragraph*{Related Work.} First efforts to formalize  
\emph{parallel} fpt algorithms are due to Cesati and Di
Ianni~\cite{CesatiI1998}, though the definitions were rather ad
hoc. Around the same time, Cai et al.~\cite{CaiCDF97} investigated
\emph{space bounded} fpt algorithms -- and since 
logarithmic space is closely related to parallel computations, these
algorithms can be seen as parallel fpt results.
A
first experimental analysis of a parallel 
fpt algorithm for vertex cover is due to Cheetham et
al.~\cite{CheethamDRST03}.
Recent work on a theoretical framework for parallel fpt has
mainly been done by Bannach et~al.~\cite{BannachST15,BannachT17} and
Elberfeld et~al.~\cite{ElberfeldST15}. These papers establish hierarchies of
parallel parameterized complexity classes and place well-known
problems in them, but do not consider
work-efficiency. Many algorithms in the cited papers are
based on the expensive \emph{color-coding} technique, which
needs work $O(n\log^2 n \log c \cdot c^{k^2} \cdot k^4 )$ and
results in unpractical algorithms.

\paragraph*{Organization of This Paper.}
Following the preliminaries, we investigate, in order, parallel search trees,
parallel kernels, and parallel interleaving. 

\section{Preliminaries}

A \emph{parameterized problem} $Q$ is a set
$Q\subseteq\Sigma^* \times \mathbb N$, where in an instance $(x,k) \in
\Sigma^* \times \mathbb N$
the number $k$ is called the \emph{parameter.} A parameterized problem
is \emph{fixed-parameter tractable} (in $\Class{FPT}$) if
there is an algorithm that decides for all $(x,k)\in\Sigma^* \times
\mathbb N$ whether $(x,k) \in Q$ holds in time  $f(k)\cdot 
|x|^c$. Here, and in the following, $f$ is always a computable
function and $c$~a constant.
As model of parallel computation we use standard \textsc{pram}s (rather
than circuits), see for instance~\cite{JaJa92}. For a
\textsc{pram} program, let $T_p(n)$ denote the maximum time the program needs
on inputs of length~$n$ when $p$ processors are available. Let $T(n) =
\inf_{p \to \infty} T_p(n)$ and let $W(n)$ denote the maximum 
number of computational steps (summed over all non-idle processors)
performed by the algorithm on inputs of length~$n$. It is
well-known that $T_p(n) \le W(n)/p +
T(n)$ holds when the set of non-idle processors is easily computable
for each step (so a compiler can schedule the to-be-done
work for each step when less processors are available than there is
work to be done)~\cite{JaJa92}. We have $T_p(n) \ge W(n)/p$
and $T_p(n) \ge T(n)$. Since for fast parallel
algorithms we have $W(n) /p >\!\!> T(n)$, the work of a parallel
algorithm is the dominating factor. We say an algorithm is
\emph{work-optimal} if its work is the best possible among all
algorithms.
This definition hinges, to a certain degree, on the
fact that there are clear notions of ``minimal work'' and ``minimal
runtime''. In the parameterized world, however, this is no longer the
case: it is not clear which of the terms $3^kn$, $2^kn^2$, $n^3+2^k$,
and $n^k$ is ``minimal.'' Depending on the values of $n$ and~$k$,
any of the terms may be more desirable than
the others. For this reason, we strive for optimality only with
respect to the following notion (throughout the paper, we
assume that functions like $W(n,k)$ or $T(n,k)$ are monotone with
respect to both parameters):
An algorithm~$A$ is \emph{work-competitive to a function~$f$} if
$W_A\in O(f)$, that is, if $W_A(n,k) \le c\cdot f(n,k)$
for all $n \ge n_0$ and $k\ge k_0$ for some constants $c$, $n_0$,
and $k_0$. An algorithm~$A$ is \emph{work-competitive to an
  algorithm~$B$} if it is work-competitive to the function~$W_B$.

\section{Work-Efficient Parallel Search Tree Algorithms}\label{sec:search-trees}

For a parameterized problem~$Q$ and an instance $(x,k)$, a
\emph{search tree} algorithm invokes a \emph{branching rule} (or
\emph{branching algorithm}) to determine a sequence $(x_1,k_1)$, \dots,
$(x_m,k_m)$ of new instances such that $(x,k) \in Q$ if, and only if,
we have $(x_i,k_i) \in Q$ for at least one~$i$. Crucially, each $k_i$
must be smaller than $k$, that is, $d_i = k - k_i > 0$.  (Let us also
require $|x_i| \le |x|$ to simplify the presentation, but this is less
crucial.) The search tree algorithm recursively calls itself on these
new instances (unless it can directly decide the
instance for ``trivial''~$k$ or for ``trivial''~$x_i$). An
example of a search tree algorithm is the branching algorithm
for the vertex cover problem where we ``branch on an arbitrary
edge'': Map $(G,k)$ to $(G-\{u\},k-1)$
and $(G-\{v\},k-1)$ for an arbitrary edge $\{u,v\}$ (we have
$d_1 = d_2 = 1$ and $m=2$). Another example is the
branching rule ``branch on the maximum-degree vertex and either take
it into the vertex cover or all of its neighbors,'' meaning that we map $(G,k)$ to
 $(G - \{u\}, k-1)$ and $(G - N(u), k - |N(u)|)$ where $N(u)$
is the neighborhood of~$u$. This leads to $d_1 = 1$ and
$d_2 = |N(u)|$; and since we can solve the vertex cover problem
directly in graphs of maximum degree $2$, we have $d_1 = 1$ and
$d_2 \ge 3$.

\subsection{Simple Parallel Search Trees}

As mentioned in the introduction, parallelizing a search tree
is more or less trivial, since we can process all
resulting branches in parallel. Of course, it
may now become important how well the \emph{branching rule} can be
parallelized, since we have to invoke it on each level of the tree.
In detail, for a set~$D$ of vectors $d =
(d_1,\dots,d_m)$, a \emph{$D$-branching algorithm $\Algo B$ for~$Q$}
is an algorithm that on input $(x,k)$ either correctly outputs ``$(x,k) \in Q$'',
``$(x,k) \notin Q$'', or instances
$(x_1,k-d_1)$, \dots, $(x_m,k-d_m)$ for some $d\in D$ such that $(x,k)
\in Q$ if, and only if, $(x_i,k-d_i) \in Q$ for some $i \in
\{1,\dots,m\}$. Let $\Algo{SeqSearchTree-B}$ and
$\Algo{ParSearchTree-B}$ denote the sequential and parallel search
tree algorithms based on~$\Algo B$, respectively. Note that both
algorithms traverse the same tree on an input~$(x,k)$. Let 
$\operatorname{size}_{\Algo B}(n,k)$ and $\operatorname{depth}_{\Algo
  B}(n,k)$ denote the maximum number of nodes and the maximum depths
of the search trees traversed by the algorithms on inputs of
length~$n$ and parameter~$k$, respectively. 

From a sequential perspective, the objective in the design of search
tree algorithms is to reduce the size of the search tree since this
will be the dominating factor in the runtime. From the parallel 
perspective, however, we will also be interested in the \emph{depth}
of the search tree since, intuitively, this depth corresponds to the
parallel time needed by the algorithm.

\begin{theorem}\label{thm-branch-simple}
  Let $\Algo B$ be a branching algorithm. Then
  {\begin{align*}
    T_{\Algo{SeqSearchTree-B}}(n,k)=
    W_{\Algo{SeqSearchTree-B}}(n,k)&= O(\operatorname{size}_{\Algo
                                       B}(n,k) \cdot W_{\Algo B}(n,k)), \\
    T_{\Algo{ParSearchTree-B}}(n,k)&= O(\operatorname{depth}_{\Algo
                                       B}(n,k) \cdot T_{\Algo B}(n,k)),\\
    W_{\Algo{ParSearchTree-B}}(n,k)&= O(\operatorname{size}_{\Algo
                                       B}(n,k) \cdot W_{\Algo B}(n,k)). 
  \end{align*}}
\end{theorem}
\begin{proof}
  This follows directly from the definitions. Note that the runtime of
  a sequential simulation of a parallel algorithm $\Algo B$ takes
  time~$W_{\Algo B}(n,k)$ and if $\Algo B$ is already a sequential
  algorithm, then $T_{\Algo B}(n,k) = W_{\Algo B}(n,k)$.
  \qed
\end{proof}

Of course, a lot is known concerning the size of search trees
resulting from $D$-branching algorithms: If
$s(k) = \operatorname{size}_{\Algo B}(n,k)$ is independent of~$n$, we
always have $s(k) \le \max_{(d_1,\dots,d_m)\in D} (s(k-d_1) + \cdots +
s(k-d_m)+1)$ and it is known~\cite{NiedermeierR2000} how to 
compute a number $\xi_D$ such that $s(k) = \Theta(\xi_D^k)$ is a
minimal solution of the inequality: for $d = (d_1,\dots,d_m)$ the
number $\xi_d$ is the reciprocal of the minimal root of the polynomial
$p(x) = 1- \sum_{i=1}^m x^{d_i}$ and $\xi_D = \sup_{d\in D} \xi_d$. For  instance, for the simple
branching algorithm for the vertex cover problem with $D = \{(1,1)\}$ we
have $\xi_D = 2$ and the search tree has size $2^k$, while for  $D =
\{(1,3); (1,4); (1,5); \dots \}$  from the  branch-on-a-degree-3-vertex
algorithm we have $\xi_D = \xi_{(1,3)} \approx 1.4656$.
Regarding the depth of the search tree, it is clearly upper-bounded by
$k/\min d$ for the ``worst $d \in D$'' since in each recursive call we
decrease $k$ by at least the minimal $d_i$ in~$d$. In summary, we see that
$\Algo{ParSearchTree-B}$ is always work-competitive to 
$\Algo{SeqSearchTree-B}$ and $T_{\Algo{ParSearchTree-B}}(n,k) =
\smash{\frac{k}{\max_{d\in D}\min_i d_i}} \cdot T_{\Algo B}(n,k)$ and
$W_{\Algo{ParSearchTree-B}}(n,k)= \xi_D^k \cdot W_{\Algo
  B}(n,k)$.

\subsection{Shallow Parallel Search Trees}\label{sec:shallow-st}

If we wish to find faster work-optimal parallel search tree algorithms,
a closer look at Theorem~\ref{thm-branch-simple} shows that there are
two lines of attack: First, we can try to decrease $T_{\Algo B}(n,k)$
while keeping $W_{\Algo B}(n,k)$ optimal. Second, we can try to
decrease the depth of the search trees without increasing their size.

Regarding the first line of attack, there is often ``little that we
can do'' since $T_{\Algo B}(n,k)$ will often already be optimal. For
instance, the branching algorithm ``pick an
arbitrary edge'' can be implemented optimally in parallel time
$O(1)$ assuming an appropriate memory access model; and for the
branch-on-a-degree-3-vertex algorithm, both finding a degree-3 vertex and
solving the instance if no such vertex exists can be done work-optimally
in polylogarithmic time.

Regarding the second line, however, new algorithmic ideas \emph{are}
possible and lead to work-efficient algorithms whose runtime is
\emph{logarithmic} in the parameter instead of linear. A word of
caution, however, before we proceed: We improve runtimes from $O(k +
\log^{O(1)} n)$ to $O(\log k + \log^{O(1)} n)$, where the $O(\log n)$
is needed already for many pre- and postprocessing operations on the
input. Clearly, the improvement in the runtime is rather modest
since we generally think of $k$ being something very small.
Nevertheless, achieving even this modest
speedup optimally is highly nontrivial for many problems.

To get some intuition for the idea, consider once more the vertex
cover problem, but let us now try to 
find \emph{ten} arbitrary edges that form a matching. Then every
vertex cover of the 
input graph must contain at least one endpoint from each of these ten
edges and we get the following new branching rule: Branch to all 1024
possible ways of choosing one vertex from each of the ten edges, each
time reducing the size of sought vertex cover by 10. This corresponds
to a branching vector $d' = (10,10,\dots,10)$ of length~1024; compared
to the vector $d = (1,1)$ if we branch over a single 
edge. In the sequential setting this idea only complicates
things since $\xi_{d'} = \xi_d = 2$ and this new algorithm
produces a search tree of the same size as before. In contrast, in the
parallel setting we make progress as the \emph{depth} of the search
tree is decreased by a factor of~10, without an increase in the work
being done. Naturally, a factor-10 speedup is just a
constant speedup, but we can extend the idea to move from a runtime
of~$k$ to $\log k$: 

\begin{theorem}\label{thm-vc-log-time}
  There is an algorithm that solves $\PLang{vertex-cover}$
  in time $T(n,k) = O(\log k \cdot \log^3 n)$ and work $W(n,k)
  = O(2^k n)$ on a \textsc{crcw-pram}.
\end{theorem}

\begin{proof}
  On input $(G,k)$ we determine a \emph{maximal} matching~$M$
  in~$G$. Clearly, if $|M| > k$, then no vertex cover of size~$k$ is
  possible and we can just ouput ``$(G,k) \notin
  \Lang{vertex-cover}$''; and if $|M| \le k/2$, then the endpoints of
  the edges in~$M$ form a vertex cover of size at most 
  $2|M| \le k$, again allowing us to  stop immediately. The interesting
  case is thus $k/2 \le |M| \le k$ and, here, we branch over all
  $2^{|M|}$ possible ways in which we can chose one endpoint from each
  edge. In the worst case, this gives a branching vector $d =
  (k/2,\dots,k/2)$ of length $2^{k/2}$. In particular, in each branching
  step we reduce the target size of the vertex cover by at least 50\%
  and, thus, after at most $O(\log k)$ steps we arrive at a trivial
  instance. The size of the search tree is not affected and, thus, still
  has size $2^k$. Since it is known~\cite{Han1995} that maximal
  matchings can be computed in time $O(\log^3n)$ and linear work, we
  obtain the claim.
  \qed
\end{proof}

The above theorem and its proof transformed a simple ``original'' search 
tree for the vertex cover problem into a ``highly parallel'' one. The
key concepts behind this transformation were the following:
\begin{itemize}
\item
  \emph{Branch structures:}
  The original branching algorithm first found ``a substructure on
  which to branch.'' For example, the vertex cover branching
  algorithm normally finds ``an arbitrary edge;'' the 
  branch-on-degree-at-least-3 algorithm finds ``a high-degree
  vertex.''
\item
  \emph{Conflict-free branch structures:}
  If the original branching algorithm has the choice among several
  possible substructures on which it could branch and if the
  substructures are disjoint, we can also branch on these 
  structures ``in parallel.'' In Theorem~\ref{thm-vc-log-time},
  ``disjoint substructures that are edges'' are matchings and we can
  branch  on them in parallel; for the branch-on-degree-at-least-3
  algorithm we can branch in parallel on any star forest.
\item
  \emph{A large number of conflict-free branch structures:}
  Lastly, we need to be able to find a large enough collection of such
  disjoint substructures quickly and work-efficiently. Its size needs to
  be at least a fraction of the parameter to ensure that we get a
  depth that is logarithmic in the parameter.
\end{itemize}
Since formalizing the above notions can easily lead to rather
technical definitions, we suggest a formalization that is not as
general as it could be, but that nicely captures the essential ideas.
We only consider
\emph{vertex search problems~$Q$} on simple graphs
where the objective is to find a parameter-sized subset of the vertices
that has a certain property.
Concerning branching rules, we only consider rules that identify a
subset of the vertices and then branch over different
ways in which some of these vertices can be added to the partial
solution:

\begin{definition}[Local branching rule]
  Let $Q$ be a vertex search problem. A \emph{local branching rule} is
  a \emph{partial} mapping that gets a tuple as input consisting of a
  graph $G = (V,E)$, a parameter~$k$, an already computed partial
  solution $P \subseteq V$, and a set $S \subseteq V \setminus P$ on which we
  would like to branch. If defined, it outputs a family~$F$ of
  nonempty subsets of~$S$ such that for every solution $Y \supseteq P$ for
  $(G,k)$ the intersection $Y \cap S$ is a superset of an element of $F$.
\end{definition}

The local branching rule for the vertex cover
algorithm maps the tuple $(G,k,P,\{u,v\})$ with $\{u,v\} \in E$ and
$u,v \notin P$ to $\{\{u\}, \{v\}\}$ and is undefined otherwise.
For the branch-on-degree-at-least-3 rule, if $S$ is
the closed neighborhood in $G - P$ of some vertex~$v$ of degree $3$ in
$G-P$, we map $(G,k,P,S)$ to $\{\{v\}, S \setminus \{v\}\}$.
Returning to the three ingredients of the proof of
Theorem~\ref{thm-vc-log-time}, the sets~$S$ in the definition of a
local branching rule are exactly the sought ``branching structures.''
A collection $M$ of such sets is ``conflict-free'' if all members of
$M$ are pairwise disjoint. In the proof of
Theorem~\ref{thm-vc-log-time} such an $M$ was simply a matching in the
graph; but given any collection $N$ of sets~$S$, any maximal set
packing $M \subseteq N$ will be conflict-free.  Maximal set packings
can be obtained efficiently and quickly in parallel by building a
conflict graph over the sets and computing a maximal independent
set \cite{KarpW85}. Therefore, in a general setting it suffices to compute a
polynomial-size set $N$ of sets~$S$ that has a set packing
$M \subseteq N$ whose size at least a fraction
of~$k$. Algorithm~\ref{alg:b1bstar} makes these ideas precise.

\begin{definition}
  An \emph{implementation} of a local branching rule consists of three
  algorithms $\Algo{decide}$, $\Algo{choices}$, and $\Algo{branches}$
  with the following properties:
  \begin{enumerate}
  \item On inputs $(G,k,P)$ for which there is no $S$ such that the
    local branching rule is defined for $(G,k,P,S)$, algorithm
    $\Algo{decide}$ must correctly output ``yes'' or ``no'' depending
    on whether $P$ is a partial solution.
  \item For all other inputs $(G,k,P)$, the algorithm $\Algo{choices}$
    must output a nonempty set $N$ such the local branching rule is
    defined on all $(G,k,P,S)$ for $S \in N$.
  \item For all $(G,k,P,S)$ for which the local branching rule is
    defined, $\Algo{branches}$ must output the corresponding
    family~$F$ of branches.
  \end{enumerate}
\end{definition}

\begin{lstlisting}[
  style=pseudocode-bannachtantauskambath,
  label=alg:b1bstar,
  captionpos=t,
  float=htpb,
  caption={
    For an implementation
    $(\Algo{decide},\Algo{choices},\Algo{branches})$, $\Algo B_1$ is the resulting
    standard branching rule. The new parallel branch
    algorithm $\Algo B_*$ first computes a set packing~$M$ of the
    set~$N$ of possible branch structures and then branches on  
    all of them simultaneously. Let $s$ be the maximum size of any~$X$
    produced in $\Algo B_1$ on any input.
  }
  ]
algorithm $\Algo B_1(G,k,P)$
  if $\Algo{decide}(G,k,P) \in \{\mathrm{yes}, \mathrm{no}\}$ then return $\Algo{decide}(G,k,P)$
  $N$ <- $\Algo{choices}(G,k,P)$  // for vertex cover, $N$ is the edge set of $G-P$
  $S$ <- an arbitrary element of $N$  // for vertex cover, $S = \{u,v\}$ for some edge in $N$
  for each $X \in \Algo{branches}(G,k,P,S)$ par do
    output in parallel $(G,k,P\cup X)$

algorithm $\Algo B_*(G,k,P)$
  if $\Algo{decide}(G,k,P) \in \{\mathrm{yes}, \mathrm{no}\}$ then return $\Algo{decide}(G,k,P)$    // Recursion break
  $N$ <- $\Algo{choices}(G,k,P)$  // for vertex cover, $N$ is the edge set of $G-P$
  $M$ <- a maximal set packing of $N$ among those of size at most $(k-|P|)/(s+1)$
  $\{S_1,\dots,S_m\}$ <- $M$ // name the elements of $M$
  for each $X_1 \in \Algo{branches}(G,k,P,S_1)$, $\dots$, $X_m \in \Algo{branches}(G,k,P,S_m)$ par do
    output in parallel $(G,k,P\cup X_1 \cup \cdots \cup X_m)$
\end{lstlisting}

\begin{theorem}\label{thm-local-log-time}
  Given an implementation
  $(\Algo{decide},\Algo{choices},\Algo{branches})$ for a local
  branching rule for some~$Q$, algorithms $\Algo
  B_1$ and $\Algo B_*$ from Algorithm~\ref{alg:b1bstar} satisfy:
  {\begin{enumerate}
  \item $\Algo{ParSearchTree-B}_*$ is work-competitive to
    $\Algo{SeqSearchTree-B}_1$ if $W_{\Algo{decide}}(n,k) +
    W_{\Algo{choices}}(n,k) + W_{ \Algo{branches}}(n,k) \in\Omega(n^3)$.
  \item If the size of the maximal set packings~$M$ computed by $\Algo
    B_*$ is always at least
    $\varepsilon(k - |P|)$ for some $\varepsilon>0$, then
    $T_{\Algo{ParSearchTree-B}_*}(n,k) = O(\log k \cdot
    (T_{\Algo{decide}}(n,k) + T_{\Algo{choices}}(n,k) + T_{
     \Algo{branches}}(n,k) + \log^4 n))$.
  \end{enumerate}}
\end{theorem}
\begin{proof}
  To see that $\Algo{ParSearchTree-B}_*$ is work-competitive to
  $\Algo{SeqSearchTree-B}_1$ first note that both algorithms produce
  search trees of the same size (albeit different depths) since each
  parallel branching done by $\Algo B_*$ over $S_1,\dots,S_m$
  corresponds to a sequential branching of $\Algo B_1$ over the same
  sets in an arbitrary order. At this point it is important that in
  $\Algo B_*$ we restrict the size of $M$ to $(k-|P|)/(s+1)
  <(k-|P|)/s$ -- otherwise, $\Algo{SeqSearchTree-B}_1$ might
  ``immediately notice after one branching'' that an input like
  $(G,2,\emptyset)$ does not have a solution, while $\Algo B_*$ might
  output a huge set~$M$ and would then branch in a great number of
  ways, only to notice immediately in each branch that no solution
  results. For instance, if $G$ is a size-1000 matching, then
  $\Algo{SeqSearchTree-B}_1$ would notice after one branching that
  $(G,2,\emptyset)$ has no solution, while an unrestricted maximal
  matching in~$G$ obviously has size $1000$ and
  $\Algo{ParSearchTree-B}_*$ would branch in $2^{1000}$ ways, each
  time immediately noticing that the solution is 998 vertices too
  large.

  Additional work inside the algorithm $\Algo B_*$ is caused by the
  need to compute a maximum set packing. This can be done by
  constructing a conflict graph, which requires work $O(n^3)$, and
  then applying the parallel maximal independent set algorithm by
  Karp and Wigderson~\cite{KarpW85}, which requires work $O(n^2)$. 

  Concerning the runtime, note that by assumption each time $\Algo
  B_*$ calls itself recursively, the size of $k - |P|$ is shrunk by at
  least a factor of $\varepsilon$. Thus, starting with $P =
  \emptyset$, after $O(\log k)$ rounds we will have $|P| = k$ and no
  further branching will happen. This immediately gives us the claimed
  runtime since computing maximal set packings can be done in time
  $O(\log^4 n)$, see~\cite{KarpW85}. 
  \qed
\end{proof}

\section{Work-Efficient Parallel Kernels}\label{sec:kernels}

Kernels are self-reductions that map instances to new
instances whose size is bounded in terms of the parameter.
Like search trees, they are basic concepts of fpt theory.
Unlike search trees, kernels are often hard to parallelize: They are
typically described in terms of \emph{reduction rules,} which
locally change an input instance in such a way that it gets a bit
smaller without changing problem membership and such that at least
one rule is still applicable as long as the instance size is not
bounded in terms of the parameter. Unfortunately, it is known that some
sets of reduction rules are ``inherently sequential,'' meaning that
computing the result of applying them exhaustively is complete for
sequential polynomial time~\cite{BannachT2018b}. On the other hand,
some reduction rules can easily be applied in parallel just as well as
sequentially, leading to kernelization algorithms running in
polylogarithmic time or even in constant time~\cite{BannachT2018a}.

While it seems hard to characterize which sets of reduction rules 
yield parallel kernels, the situation is
more favorable when we consider a sequence of kernels (a
\emph{kernel cascade}). In the sequential setting, the situation is
simple: Given several kernelizations for the same problem, the
asymptotically fastest way to compute a minimum-size kernel is simply
to apply them in sequence starting with the fastest and ending
with the slowest. In the parallel setting, the situation is also
simple when we can parallelize all kernels of a cascade
optimally. However, even when this is not the case, we may still 
get a fast parallel algorithm and there is an intriguing dependence on
the parallel runtime and the kernel size: Theorem~\ref{thm-parallel-kernels} states
that it suffices to parallelize the kernels in a cascade until the
\emph{kernel size} equals the desired \emph{parallel runtime} -- while
later kernels 
need not be parallelized. 

\subsection{Sequential Kernel Cascades}

A \emph{kernelization} for a parameterized problem $Q \subseteq \Sigma^*
\times \mathbb N$ is a polynomial-time computable function $K \colon \Sigma^* \times \mathbb N
\to \Sigma^* \times \mathbb N$ such that (a)
$(x,k) \in Q$ if, and only if,  $K(x,k) \in Q$ for the
\emph{kernel $K(x,k)$} and such that (b) for some computable
function~$s_K$ we  
have $|K(x,k)| \le s_K(k)$ for all $x$ and~$k$. We
call the kernelization \emph{polynomial} if $s_K$ is a
polynomial.  A \emph{kernel algorithm} is an algorithm~$\Algo K$ that
computes a kernelization~$K$.

As indicated earlier, there can be several kernelizations (and, hence,
kernel algorithms) for the
same problem and they may differ regarding their runtime and their kernel
sizes. For instance, on input $(G,k)$ the \emph{Buss kernelization} of
the vertex cover problem removes all vertices of degree larger
than~$k$ (which must be in a vertex cover) and then removes all
isolated vertices (which are not 
needed for a vertex cover). It yields kernels of size
$s_{\Algo{Buss}}(k) = k^2$ and can be computed very quickly. In
contrast, the \emph{linear program kernelization}~\cite{ChenKJ2001}
for the vertex cover 
problem solves a linear program in order to compute a
kernel of size~$2k$, but solving the linear program takes more time.
It now makes sense to \emph{first} compute a Buss kernel 
\emph{followed} by an application of the linear program kernelization since
we then apply a ``slow'' algorithm only to an already reduced
input size (from originally $n$ to only $k^2$).

In general, let a \emph{kernel cascade}
be a sequence $C = (\Algo K_1, \dots, \Algo K_t)$ of
kernel algorithms for the same parameterized problem~$Q$
sorted in strictly increasing order of runtime (that is, we require
$T_{\Algo   K_i} \in o(T_{\Algo K_{i+1}})$ and thereby implicitly
rule out situations where runtimes are incomparable) and strictly
decreasing order of kernel sizes (that is, we require $s_{\Algo K_i}
(k) > s_{\Algo K_{i+1}}(k)$ for all but finitely many~$k$). The
\emph{cascaded kernel algorithm~$\Algo K_C$ of a cascade~$C$} will, on
input $(x,k)$, first apply $\Algo K_1$ to $(x,k)$, then applies $\Algo
K_2$ to the result, then $\Algo K_3$ and so on, and output the result
of the last $\Algo K_t$. Clearly, the following holds: 

\begin{observation}\label{obs-optimal-kernel-cascade}
  Let $C = (\Algo K_1, \dots, \Algo K_t)$ be an kernel
  cascade. Then $s_{\Algo K_C} = s_{\Algo K_t}$ and the runtime of
  $\Algo K_C$ is
  \(
    T_{\Algo K_C}(n,k) = T_{\Algo K_1}(n) + T_{\Algo K_2}(s_{\Algo
    K_1}(k)) + T_{\Algo K_3}(s_{\Algo K_2}(k)) + \cdots +
    T_{\Algo K_t}(s_{\Algo K_{t-1}}(k)).
    \)
  Furthermore, no subsequence~$C'$ of~$C$ with $s_{\Algo
    K_{C'}} = s_{\Algo K_t}$ achieves an asymptotically faster
  runtime. 
\end{observation}

\subsection{Parallel Kernel Cascades}

Faced with the problem that kernels based on reduction rules are often
difficult to parallelize, parallelizing a whole kernel cascade in a
work-optimal way seems even more challenging:
Observation~\ref{obs-optimal-kernel-cascade} states that for a given 
cascade the asymptotically fastest runtime is achieved by applying all
kernels in the cascade in sequence. Since ``work optimal'' means, by
definition, ``parallel work equal to the fastest sequential runtime,''
we also must apply work-optimal parallel versions of all kernels in
the cascade in sequence in the parallel setting.

It turns out that it may \emph{not} be necessary to
parallelize all kernels in a cascade: Suppose we only
parallelize the first kernel in a cascade, that is, suppose we find a
work-optimal algorithm for $\Algo K_1$ with runtime $O(\log^{O(1)} n)$
and then apply this parallel algorithm \emph{followed by the
  unchanged sequential kernels $\Algo K_2$ to~$\Algo K_t$.} The work
of the resulting cascade will be identical to the runtime of the
original sequential cascade (since $\Algo K_1$ is work-optimal and
nothing else is changed). The runtime, however, will now be
$O(\log^{O(1)} n)$ plus some function \emph{that depends only on~$k$}
(since all later kernels are applied to inputs whose size depends only
on~$k$). Assuming that we consider a runtime of the form
$O(\log^{O(1)} n + f(k))$ ``acceptable,'' we see that we can turn any
sequential kernel cascade into a parallel one by parallelizing only
the \emph{first} kernel.
Of course, there are functions~$f$ that we might not consider
``acceptable''; for instance, $f$ might be exponential. Intuitively, we
then need to ``parallelize more kernels of the sequence.''

\begin{theorem}\label{thm-parallel-kernels}
  Let $C = (\Algo K_1, \dots, \Algo K_t)$ be a kernel
  cascade and for some $r \le t$ let $\Algo K'_1,\dots, \Algo K'_r$ be
  parallel implementations of $\Algo K_1,\dots, \Algo K_r$, that is,
  for $i \in \{1,\dots,r\}$ let $\Algo K'_i$ be a work-competitive
  parallel implementation of~$\Algo K_i$ with runtime $T_{\Algo K'_i}
  \in O(\log^{O(1)} n)$. Let $C' = (\Algo K'_1, \dots, \Algo K'_r,
  \Algo K_{r+1},\dots, \Algo K_t)$. Then
  {\begin{enumerate}
  \item $\Algo K_{C'}$ is work-competitive to $\Algo K_C$ and
  \item $T_{\Algo K_{C'}}(n,k) = \log^{O(1)} n + s_{\Algo K'_r}(k)^{O(1)}$.
  \end{enumerate}}
\end{theorem}
\begin{proof}
  Consider arbitrary kernel cascades $C$ and $C'$ defined as
  above. Since $\Algo K'_i$ is work-competitive to $\Algo K_i$ and they
  have asymptotically the same kernel size $s_{K'_i}\in O(s_{K_i})$ for
  every $i\in\{1,\dots,r\}$, it follows by
  Observation~\ref{obs-optimal-kernel-cascade} that the cascade
  $C'_\Delta=(K'_1,\dots,K'_r)$ is work-competitive to the cascade
  $C_\Delta=(K_1,\dots,K_r)$. Furthermore, $s_{C'_\Delta}\in O(s_{C_\Delta})$.
  In both cascades $C$ and $C'$, the remaining phases after computing a
  similar-sized kernel by $C_\Delta$ and $C'_\Delta$ are equal. Since
  in both cases the input for this phase has similar size, we get that
  $C'$ is work-competitive to~$C$.

  In~$C'$, for every $i\leq r$ the kernel $K'_i$ needs parallel time
  $\smash{T_{K'_i}(n)=\log^{O(1)}n}$. It follows directly that $T_{C'_\Delta}(n)\in
  O(\log^{O(1)} n)$. The output of the first stage has size at most
  $s_{K'_r}(k)$. The dominating work in the remaining phase of the
  cascade $C'$ is the polynomial work of the last kernel
  algorithm. Since this work is polynomial in $n$, we directly get a
  maximal work of $s_{K'_r}(k)^{O(1)}$, which completes the proof.
  \qed
\end{proof}

As a concluding example, consider once more $\PLang{vertex-cover}$. We
mentioned already that there is a size-$k^2$ kernel algorithm
$\Algo{Buss}$ for this problem, which is 
easy to implement in linear sequential time, but also 
in logarithmic parallel time and linear work (and,
thus, optimally). There is also a size-$2k$ kernel algorithm
$\Algo{LP}$ based on~\cite{ChenKJ2001} that needs sequential time
$O(|E|\sqrt{|V|})$. For this kernel, no work-optimal (deterministic)
polylogarithmic time implementation is known (indeed, \emph{any}
parallel implementation is difficult to achieve~\cite{BannachT2018b}).  
By Observation~\ref{obs-optimal-kernel-cascade}, there is a
\emph{sequential} kernel algorithm for the vertex cover problem that
runs in time $O\bigl(n + k^2 \sqrt{k^2}\bigr) = O(n+k^3)$. By
Theorem~\ref{thm-parallel-kernels}, there is a \emph{parallel} kernel
algorithm that is work-competitive and needs time $O(\log n + k^3)$.  

\section{Work-Efficient Parallel Interleaving}\label{sec:interleaving}

\emph{Interleaving} is a method to combine a branching
algorithm~$\Algo B$ and kernel algorithm~$\Algo K$ 
to ``automatically'' reduce the runtime of
$\Algo{SeqSearchTree-B}$: During the recursion, the algorithm
$\Algo{SeqInterleave-B-K}$ applies $\Algo K$ at the beginning of
each recursive call (thus, calls to the kernel algorithm are
``interleaved'' with the recursive calls, hence the name of the
method). Intuitively, at the start of the recursion, calling a kernel
algorithm is superfluous (the input is typically already
kernelized) and only adds to the runtime, but deeper in the recursion
it will ensure that the inputs are kept small. Since the bulk of
all calls are ``deep inside the recursion'' we can hope that ``keeping
things small there'' has more of a positive effect than the negative
effect caused by the superfluous calls at the beginning. Niedermeier
and Rossmanith have shown that this intuition is correct:

\begin{fact}[\cite{NiedermeierR2000}]\label{fact-seq-interleaving}
  Let $\Algo K$ be an arbitrary kernel algorithm that
  produces kernels of polynomial size. Let $\Algo B$ be a
  $d$-branching algorithm running in polynomial time. Then
  \(
    T_{\Algo{SeqInterleave-B-K}}(n,k)= \operatorname{size}_{\Algo
    B}(n,k) + n^{O(1)} \le \xi_d^k + n^{O(1)}.
  \)
\end{fact}

\subsection{Simple Parallel Interleaving}

Interleaving also helps to reduce the work of parallel search tree
algorithms: Consider the algorithm
$\Algo{ParInterleave-B-K}$, the version of $\Algo{ParSearchTree-B}$
that applies $\Algo K$ at the beginning of each recursive call.
First applying $\Algo K$ and then computing branch
instances using $\Algo B$ is itself a branching algorithm and, thus,
Theorem~\ref{thm-branch-simple}  tells us that
$T_{\Algo{SeqInterleave-B-K}}(n,k)=W_{\Algo{ParInterleave-B-K}}(n,k)$
holds.
This observation suggests that in order to minimise the work, we
have to choose the most work-efficient kernel algorithm $\Algo K$
available to us. However, it turns out that 
we have more options in the parallel setting: The work of $\Algo K$ is
only relevant at the 
very beginning, when the input size still depends on $n$. Later on,
all remaining computations get inputs whose size depends only on the
parameter. For these calls, the work of $\Algo K$ is no longer
relevant -- it is ``drowned out'' by $\xi_D^k$. This suggests the
following strategy: We use \emph{two} kernels, namely an \emph{initial
kernel} whose job is to quickly and, more importantly, work-efficiently
reduce the input size once (how such kernels can be constructed was
exactly what we investigated in Section \ref{sec:kernels}); and then use an
\emph{interleaving kernel} during the actual interleaving, whose job
is just to kernelize the intermediate instances as quickly as possible
-- but we need no longer care about the work! Let us write $\Algo A |
\Algo B$ for the 
sequential concatenation of algorithms $\Algo A$ and $\Algo B$.

\begin{theorem}\label{theorem:interleave}
  Let $\Algo B$ be a $d$-branching algorithm, and let $\Algo
  {K}_{\Algo{init}}, \Algo K_{\Algo{interleave}}$ be poly\-no\-mial-sized
  kernels.  Then $W_{\Algo
    K_{\Algo{init}}|\Algo{ParInterleave-B-K$_{\text{inteleave}}$}}(n,k)\in
  O(W_{\Algo K_{\Algo{init}}}(n,k)+\xi_d^k)$.
\end{theorem}
\begin{proof}
  Let $\Algo B,\Algo K_{\text{init}},$ and $\Algo
  K_{\text{interleave}}$ be algorithms as defined above. It is easy to
  see that
  \begin{align*}
    &W_{\Algo
    K_{\text{init}}|\Algo{ParInterleave-B-K$_{\text{interleave}}$}}(n,k)\\
    \leq{}&
    W_{\Algo K_{\text{init}}}(n,k) +
    W_{\Algo{ParInterleave-B-K$_\text{interleave}$}}(s_{\Algo
    K_{\text{init}}}(k),k).
  \end{align*}
  From Fact~\ref{fact-seq-interleaving} we get
  \begin{align*}
    W_{\Algo
    K_{\text{init}}|\Algo{ParInterleave-B-K$_\text{interleave}$}}(n,k)\leq
    W_{\Algo K_{\text{init}}}(n,k)+(s_{\Algo
    K_\text{init}}(k))^{O(1)}+\xi_{d}^k
  \end{align*}
  and finally with
  $(s_{\Algo K_\text{init}}(k))^{O(1)}\in O(k^{O(1)})$ we get the
  claim.
  \qed
\end{proof}

\subsection{Shallow Parallel Interleaving}

At the end of Section \ref{sec:search-trees} we introduced the idea of
\emph{shallow} search trees as a method to speedup parallel
search tree algorithms. However, shallow search trees are
not necessarily compatible with the interleaving technique: From the
parallel point of view, a ``perfect'' branching algorithm would branch
on input $(G,k)$ in constant time to $m = \xi_D^k$ simple
instances $(G_1,1), \dots, (G_m,1)$, all of which can then be
processed in parallel. Applying a kernel at this point is ``too
late'': The work will be something like $m = \xi_D^k$ times the work of
the kernel, which is decidedly not of the form 
$\xi_D^k$ \emph{plus} the work of the kernel. 

What goes wrong here is, of course, that we parallelize ``too much'':
we must ensure that the kernel
algorithm gets a chance to kick in while the inputs still have a
\emph{large} enough size. On inputs of (still) large parameter~$k$,
all branches have to have a parameter of size \emph{at least}
$\varepsilon k$ (normally, we want a parameter \emph{at most}
$\varepsilon k$). We remark that it does 
not follow from~\cite{NiedermeierR2000} that interleaving is possible
here since \cite{NiedermeierR2000} considers only the case where the
number of branch instances is bounded by a constant. 
For the following theorem, let us write $d(x,k)$
for the branching vector $d$ used by $\Algo B$ on input $(x,k)$ and
$|d(x,k)|$ for its length.

\begin{theorem}\label{theorem:shallowInterleaving}
  Let $\Algo B$ be a $D$-branching algorithm such that for all inputs
  $(x,k)$, (a) the work 
  done by $\Algo B$ is at most $|d(x,k)|\cdot |x|^{O(1)}$ and (b) the maximum
  value in $d(x,k)$ is at most $(1-\varepsilon)k + O(1)$.  Let $\Algo K$ be a
  polynomially-sized kernel algorithm. Then
  $W_{\Algo{ParInterleave-B-K}}(n,k)=O(W_{\Algo  K}(n,k)+\xi_D^k)$.
\end{theorem}
\begin{proof}
  Let $\Algo B$ be a $D$-branching algorithm as in the
  theorem. We first show that the total work of a shallow search tree algorithm
  can be mapped to nodes of a tree such that every node is labelled with a 
  polynomial work $p(k)$ instead of some superpolynomial work
  $\left|d(x,k)\right|k^{O(1)}$. Then we show that we get the same upper bound
  for the total work in this tree according to arguments from~\cite{NiedermeierR2000}.

  Consider the tree of $\Algo{ParInterleave-B-K}$ for an arbitrary
  input. Each node in this tree is an instance $x$ with its parameter
  $k$. $\Algo{ParInterleave-B-K}$ runs the kernel
  algorithms~$\Algo{K}$ and the branching algorithm $\Algo{B}$ on
  these instances. Let us label each node in the traversed tree with
  the work $W_1(x,k)$ that is necessary to run $\Algo{B}|\Algo{K}$ at this point,
  and let us use $T_1$ to denote this tree.  For $T_1$ there
  is an upper bound $u(k)=k/\varepsilon + O(1)$ such that for each
  node $(x,k)$ and its parent node $(x',k')$,
  we have that $k'\leq u(k)$. The sole exception is the root of
  $T_1$. This holds because every value in
  $d(x,k')$ is at most $(1-\varepsilon)k'+O(1)$ such that
  $k\geq k'-(1-\varepsilon)k'+O(1)$.  Since $\Algo K$ is a polynomial
  kernel algorithm, each node with an instance $x$ and a parameter $k$
  in $T_1$ is labelled with at most
  $W_{\Algo K}(|x|,u(k)) + |d(x,k)|\cdot q(k)$ for some polynomial~$q$. Note that the total work of $\Algo{ParInterleave-B-K}$ is the
  sum of the labels over all nodes of $T_1$. We can construct
  another labelled tree $T_2$ such that the sum of its labels is more
  than the total work of $T_1$ and such that the labels are bounded by a
  polynomial: Let $T_2$ consists of the same nodes as $T_1$, and
  let each node $(x,k)$ be labelled with
  $W_2(x,k)=W_{\Algo K}(|x|,u(k))+2q(u(k))$.  Since $W_{\Algo K}$,
  $s_{\Algo K}$, $q$, and $u$ are polynomials, $W_2$ is also a
  polynomial. We show that the total sum of all labels in
  $T_1$ is smaller than the total sum of labels in $T_2$. We can
  ignore the term $W_{\Algo K}(|x|,u(k))$ because it is part of both
  labels.  Consider some node $(x,k)$ in $T_1$ and in $T_2$: If
  $(x,k)$ has no children, then it is easy to see that
  $W_1(x,k)-W_{\Algo K}(|x|,u(k))\leq W_{\Algo B}(s_{\Algo K}(k),k)=
  q(k)\leq q(u(k))=\frac{W_2(x,k)-W_{\Algo K}(|x|,u(k))}{2}$.
  Otherwise, we have $|d(x,k)|$ many children
  $(x_1,k_1),\dots,(x_m,k_m)$, and for every $(x_i,k_i)$ we know that
  $k_i\geq \varepsilon k$. It follows that
  \begin{align*}
    |d(x,k)|\cdot q(k) \leq |d(x,k)|\cdot q(u(\varepsilon k))\leq \sum_{i=1}^m q(u(k_i)).
  \end{align*}
  As a result, the work of $(x,k)$ in $T_1$ can be mapped to the work of the children in $T_2$.
  
  Since $\Algo K$ is a polynomial kernel algorithm, for every node $(x,k)$ in $T_2$ it holds that $|x|\leq s_{\Algo K}(u(k))$, or
  $(x,k)$ is the root node.  This implies that $T_2$ is a tree in which every label is
  polynomially bounded with respect to the parameter. The only exception is the
  additional term $W_{\Algo K}(|x|,k)$ in the label of the root node. 
  It follows that the total
  work of our algorithm can be estimated by using recurrence equations where the
  inhomogeneity is a polynomial. Note that this is only possible with the
  necessary upper bound for the  values in the branching vectors. 

  It remains to show that the work of the interleaving algorithm is
  $O(\xi_D^k)$ in addition to the application of the kernel, which is bounded by
  the sum of all labels of $T_2$. To prove
  this, we use the following terminology: We wish to bound values
  $W_k$ for $k \in \mathbb N$ for which we know that the following
  holds:
  \begin{align}
    \textstyle W_k \le \sum_{i=1}^{|d^k|} W_{k-d_i^k} + f_k, \label{eq-wk}
  \end{align}
  where each $d^k$ is a branching vector (having length $|d^k|$) and
  $f_k \in \mathbb N$ are numbers. Let $\xi_{d^k}$ be the reciprocal of the minimal
  root of the polynomial
  \begin{align}
    \textstyle 1 - \sum_{i=1}^{|d^k|} x^{d_i^k} \label{eq-poly}
  \end{align}
  and let $\xi = \sup_{d^k\in D} \xi_{d^k}$. Then we have $\xi_D=\xi$.

  Define values $U_k$ (for ``upper bound'') by the 
  recurrence equation
  \begin{align}
    U_k = \xi U_{k-1} + f_k \label{eq-uk}
  \end{align}
  and $U_0 = f_0$ and observe that this recursion has a unique
  solution. We prove by induction that
  \begin{align*}
    W_k \le U_k
  \end{align*}
  holds for all $k \in \mathbb N$ and all solutions~$W_k$ of
  \eqref{eq-wk}. Clearly, the claim holds for $k=0$ since $W_k \le f_k
  = U_k$ for $k=0$. For the inductive step, observe that \eqref{eq-uk}
  clearly implies $U_k \ge \xi U_{k-1}$. Thus $U_k \ge \xi^i
  U_{k-i}$ for $i \ge 1$ and, in particular, $U_{k-1} \ge \xi^{i-1}
  U_{k-i}$, which in turn is equivalent to
  \begin{align}
    U_{k-i} \le \xi^{1-i} U_{k-1}.\label{eq-ubound}
  \end{align}
  This allows us to bound $W_k$ as follows: \allowdisplaybreaks
  \begin{align*}
    W_k &\le \textstyle \sum_{i=1}^{|d^k|} W_{k-d_i^k} + f_k \tag{by \eqref{eq-wk}} \\
        &\le \textstyle \sum_{i=1}^{|d^k|} U_{k-d_i^k} + f_k \tag{by induction hypothesis} \\
        &\le \textstyle \sum_{i=1}^{|d^k|} \xi^{1-d_i^k} U_{k-1} + f_k \tag{by \eqref{eq-ubound}} \\
        &= \textstyle \xi \sum_{i=1}^{|d^k|} \xi^{-d_i^k} U_{k-1} + f_k \\
        &\le \textstyle \xi U_{k-1} + f_k \tag{by \eqref{eq-poly}}\\
        &= U_k \tag{by definition of $U_k$}
  \end{align*}

  We now know that in order to bound the runtime of the interleaving
  algorithm, it suffices to solve the recurrence equation $U_k = \xi_D U_{k-1} +
  f_k$. However, it is well-known that when $f_k$ is a polynomial,
  this has the solution $U_k = \lambda\xi_d^k + p(k)$ for some polynomial $p$ and
  some for constant $\lambda$, see~\cite{NiedermeierR2000}. Since the labels
  $W_2(x,k)$ in the tree $T_2$ are polynomially bounded by $k$, 
  the sum of all labels in $T_2$ is $W_{\Algo{K}}(n,k)+O(\xi_D^k)$. 
  Finally, this gives us the claimed
  work for our algorithm.
  \qed
\end{proof}

Note that the search trees arising from the branching rule $\Algo
B_*$ always have property (b), that is, they never ``parallelize too
well'' since we capped to size $m$ of $M$ to $(k-|P|)/(s+1)$ and, thus,
$k - |P| - |X_1| - \cdots - |X_m| \ge k - |P| - s(k-|P|)/(s+1) =
(k-|P|)/(s+1)$, meaning that we can set $\varepsilon = 1/(s+1)$. 

\section{Conclusion and Outlook}
We have begun to extend the field of parallel parameterized algorithms with
respect to work-optimality. This is a first step towards the aim of
closing the gap between theoretical parallel algorithms (which are
fast but produce massive work) and algorithms that work well in
practice. To that end we provided a framework that allows to transform
sequential search tree algorithms as well as kernelizations into
parallel algorithms that are work-efficient. Furthermore, we have
shown that combining both techniques via interleaving is still
possible in the parallel setting. There are multiple paths to extend
this line of research: It would be interesting to know if the
presented algorithms do, in fact, lead to competitive parallel
implementations. From the theory point of view, a natural next step is
to study which other fpt techniques allow work-optimal implementations.

\bibliographystyle{plain}
\bibliography{main}

\clearpage
\end{document}